\documentclass[preprint,12pt]{elsarticle}

\usepackage{amssymb}
\usepackage{multirow}
\usepackage{amsmath,amssymb,amsfonts}%
\usepackage{amsthm}%
\usepackage{algorithm}%
\usepackage{algorithmicx}%
\usepackage{algpseudocode}%

\makeatletter
\let\elsarticle@keyword\keyword
\g@addto@macro\frontmatter{\let\keyword\elsarticle@keyword}
\makeatother
\usepackage{program}

\newtheorem{lemma}{Lemma}

\newtheorem{theorem}{Theorem}

\begin{document}

\begin{frontmatter}



\title{Colored Points Traveling Salesman Problem}


\author{Saeed Asaeedi \corref{cor1}}

 \cortext[cor1]{Corresponding author}
\ead{asaeedi@kashanu.ac.ir}
\affiliation{organization={Department of Computer Science, Faculty of Mathematical Sciences, University of Kashan},%Department and Organization
            %addressline={address }, 
            city={Kashan},
            postcode={87317-53153}, 
            %state={98},
            country={I. R. Iran}}

%\author{Second}
%\affiliation[1,2]{organization={org},%Department and Organization
%            %addressline={address }, 
%            %state={},

\begin{abstract}
%In this paper we define Colored Points Traveling Salesman Problem (Colored Points TSP) as a new variant of the classical Travelling Salesman Problem (TSP) where the set of points is divided into multiple class eseach represented by a different color (or label). The objective is to find a minimum cost hamiltonian cycle $C$ such that $C$ visits all the colors and each color appears exactly once in $C$. We prove that Colored Points TSP is NP-hard by reducing the classical TSP to it. Here, we present a $\frac{2\pi r}{3}$-approximation algorithm for this problem such that $r$ is the radius of the minimum color-spanning circle of the points.

The Colored Points Traveling Salesman Problem (Colored Points TSP) is introduced in this work as a novel variation of the traditional Traveling Salesman Problem (TSP) in which the set of points is partitioned into multiple classes, each of which is represented by a distinct color (or label). The goal is to find a minimum cost cycle $C$ that visits all the colors and only makes each one appears once. This issue has various applications in the fields of transportation, goods distribution network, postal network, inspection, insurance, banking, etc. By reducing the traditional TSP to it, we can demonstrate that Colored Points TSP is NP-hard. Here, we offer a $\frac{2\pi r}{3}$-approximation algorithm to solve this issue, where $r$ denotes the radius of the points' smallest color-spanning circle. The algorithm has been implemented, executed on random datasets, and compared against the brute force method.

\end{abstract}




\end{frontmatter}


\section{Introduction}
The Traveling Salesman Problem (TSP) is a famous classical problem in the field of combinatorial optimization. 
%Given a set of points in the plane, this problem finds the minimum cost cycle that visits all the points exactly once. 
Given a set of points in the plane, this problem finds the minimal cost cycle that visits each point exactly once. TSP is investigated by Laporte in~\cite{laporte1992traveling} and some exact and approximate algorithms are overviewed on this problem.
This classic problem has numerous variants. The Chromatic Traveling-Salesmen Problem~\cite{harvey1974chromatic} involves determining the minimum number of salesmen when the total length of a salesman's tour is limited by a specified constant.

The Colorful Traveling Salesman Problem, introduced by Xiong et al.~\cite{xiong2007colorful}, is described on a connected undirected graph whose edges are colored. The objective is to find a Hamiltonian cycle on the graph with the minimum number of distinct colors. The Multicolor Traveling Salesman Problem is defined on a complete graph whose vertices are colored~\cite{tresoldi2010solving}. Its aim is to identify the optimal Hamiltonian cycle such that the number of vertices of the subtour between two consecutive vertices of the same color is bounded. The Colored Traveling Salesman Problem, as introduced by Li et al.~\cite{li2014colored}, deals with a set of colored points. This problem represents a variant of the multiple traveling salesman problem, wherein each salesman is associated with a distinct color and is permitted to visit points of the same color, along with shared points that are common to all salesmen. In~\cite{meng2021colored}, the colored traveling salesman problem is considered on the points with varying colors. The Labeled Traveling Salesman Problem~\cite{couetoux2010labeled} is defined as the problem of finding a tour on a complete graph with colored edges, aiming to maximize or minimize the number of distinct colors used.

%In this paper, we introduce Colored Points TSP as a novel variant of the traditional TSP.  We define this problem on a set of colored points in the plane. The goal is finding the shortest cycle that visits points of all different colors and visits each color exactly once. 
In this paper, we present a novel variation of the traditional TSP, called Colored Points TSP. This problem is defined on a set of colored points in the plane, with the objective being to determine the shortest cycle that visits points of all distinct colors, ensuring that each color is visited exactly once. We prove the NP-hardness of the problem and propose an $\frac{2\pi r}{3}$-approximation algorithm for its solution, where $r$ represents the radius of the smallest color-spanning circle containing the points.

The minimum spanning circle problem was introduced by Sylvester~\cite{sylvester1857question} in 1857. A linear time algorithm is presented by Megiddo~\cite{megiddo1983linear} to solve this problem. Abellanas et al.~\cite{abellanas2001farthest} present two algorithms for computing the colored version of this problem, known as the Minimum Color-Spanning Circle Problem: Given $n$ points with $k$ different colors, the objective is to discover the smallest spanning circle that contains at least one point of each color. The first algorithm operates in $O(nk)$ time following the computation of the Farthest Color Voronoi Diagram~\cite{abellanas2001farthest} of the points in $O(n^2\alpha(k)\log{k})$ time. Meanwhile, the second algorithm, which is faster for small values of $k$, runs in $O(k^3n\log{n})$.

Abellanas et al.~\cite{abellanas2001smallest} compute the smallest color-spanning objects such as axis-parallel rectangle and the narrowest strip. Recently, Acharyya et al.~\cite{acharyya2022minimum} have delved into the consideration of the minimum color-spanning circle on imprecise points.

%The rest of the paper is as follows. In Section~\ref{sec:2}, Colored Points TSP is formally defined and an exact algorithm and an approximation algorithm are presented to solve this problem. The numerical results of the presented algorithms are discussed in Section~\ref{sec:3}. Section~\ref{sec:4} concludes the paper by highlighting its achievements.

The rest of the paper is organized as follows. Section~\ref{sec:2} formally defines the Colored Points TSP, presenting both an exact and an approximation algorithm to address this problem. Section~\ref{sec:3} delves into the numerical results of the presented algorithms. Section~\ref{sec:4} concludes the paper by highlighting its achievements.

%and proved that this problem is NP-complete. Given a graph whose edges are colored, this problem finds a Hamiltonian tour with the minimum number of distinct colors.

\section{Colored Points TSP}
\label{sec:2}
Let $S=\{s_1,\dots, s_n\}$ be a set of points in the plane and $C=\{c_1,\dots,c_k\}$ be a set of distinct colors. We assume that each point $s_i$ is associated with only one color $c_j$, which we indicate as $C(s_i)=c_j$. Consequently, the points are categorized into $k$ classes, with each class containing points of the same color. We represent the number of points in the $i$'th class as $x_i$, where $i \in \{1,2,\dots, k\}$ and $\Sigma_{i=1}^{k}x_i=n$. For every $i \in \{1,2,\dots, k\}$, we consider $CS_i$ as a vector of points with the color $c_i$, i.e., $\lvert CS_i \lvert=x_i$. The main problem involves finding a polygon $P=(p_1,p_2,\dots, p_k,p_1)$ with the shortest perimeter, ensuring that $p_i \in S$ for all $ i \in \{1,2,\dots,k\}$, and $C(p_i)\neq C(p_j)$ for all $ i\neq j \in \{1,2,\dots,k\}$. We call this problem as Colored Points TSP.

The Colored Points TSP presents a variation of the classic TSP, wherein the salesman visits all the colors (not all the points) exactly once. This concept finds applications across diverse domains, including transportation, goods distribution networks, postal systems, inspections, insurance, banking, and more. Wherever one encounters a scenario involving various branches of different service providers, and the primary aim is to navigate through exactly one branch of each service provider, optimizing for the shortest path, the Colored Points TSP emerges as a significant tool. We assume that the requirement of visiting a single branch of each provider is adequate, and that each branch of a provider is indistinguishable from the others.

In Fig.~\ref{fig:1}, we observe an instance of the problem for $n=70$ and $k=7$. It is evident that this problem is more general than the traditional TSP; specifically, when $k=n$, the Colored Points TSP effectively transforms into the TSP. As demonstrated in Theorem~\ref{thm:1}, the Colored Points TSP is proven to be NP-hard.

\begin{figure}[h]%
\centering
  \includegraphics[width=0.5\linewidth]{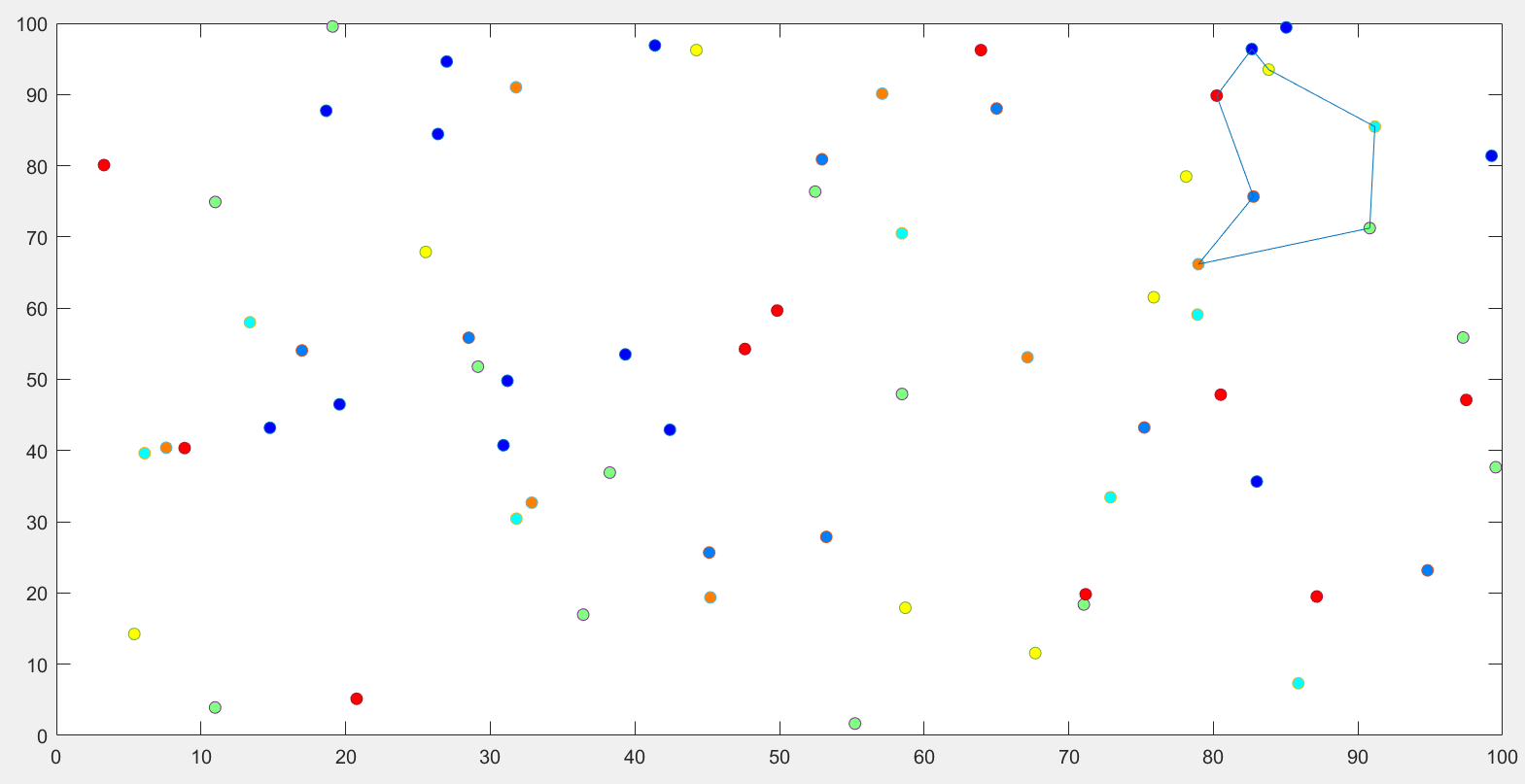}
  \caption{All 70 points are associated with one of 7 distinct colors, and the line segments collectively form a polygon with the smallest possible perimeter, encompassing all 7 unique colors.}
  \label{fig:1}
\end{figure}

\begin{theorem}
\label{thm:1}
The Colored Points TSP is NP-hard.
\end{theorem}

\begin{proof}
We reduce the clasic TSP to the Colored Points TSP by transforming each instance of TSP, denoted as $A=\{s_1,s_2,\dots, s_n\}$, into an instance of Colored Points TSP represented as $S=\{s_1,s_2,\dots, s_n\}$, $C=\{c_1,c_2,\dots,c_n\}$, and $C(s_i)=c_i$ for all $i\in \{1,2,\dots,n\}$. Notably, the polygon with the minimal perimeter traversing through all colors is equivalent to the polygon with the minimal perimeter traversing through all points..
\end{proof}

In the following, we introduce the brute-force algorithm and an approximation algorithm to address the Colored Points TSP. In the brute-force algorithm, we examine all potential combinations of $k$ distinct colors from $n$ points and determine the one with the smallest possible perimeter. For a detailed description of this approach, refer to Algorithm~\ref{alg:1}.

\begin{algorithm}
\caption{The brute-force algorithm for the Colored Points TSP}\label{alg:1}
\begin{algorithmic}
\Require $CS_i$ $\forall i \in \{1,2,\dots, k\}$ 
\Ensure $P=(p_1,p_2,\dots, p_k,p_1)$ as the polygon with the shortest perimeter that encompasses all distinct colors.
\State $P \gets ()$
\State $min \gets $ MAXINT
%\Comment{ Assume that $c_0$ is $t_i$ }
\For{$i_1 \gets 1$  \TO $x_1$ } 
{
\For{$i_2 \gets 1$  \TO $x_2$ } 
{
\State $\vdots$

\For{$i_k \gets 1$  \TO $x_k$ } 
{
\State $Q=(CS[i_1],CS[i_2],\dots, CS[i_k],CS[i_1])$
\If{Perimeter of $Q$ is less than $min$}
      \State $P \gets Q$
      \State $min \gets $ Perimeter of $Q$
\Else
\State 	Break
\EndIf 

} 
\EndFor

$\vdots$

} 
\EndFor

} 
\EndFor
\end{algorithmic}
\end{algorithm}

The time complexity of Algorithm~\ref{alg:1} is $O(x_1*x_2*\dots*x_k)$. In the worst case, when $x_1=x_2=\dots=x_k=\frac{n}{k}$, the time complexity is $O(\frac{n}{k}^k)$. In the following, we present an approximation algorithm for the Colored Points TSP, achieving a factor of $\frac{2\pi r}{3}$, where $r$ represents the radius of the smallest color-spanning circle on $S$. This approach is detailed in Algorithm~\ref{alg:2}, which outlines the steps of the approximation algorithm utilizing the onion peeling technique and the minimum color-spanning circle algorithm.

%for all $i\neq j$ in $\{1,2,\dots,k\}$.
%such that the points of each class have same color.
%such that each class contains the points with the same color.

\begin{algorithm}
\caption{The approximation algorithm for the Colored Points TSP}\label{alg:2}
\begin{algorithmic}
\Require $S=\{s_1,s_2,\dots, s_n\}$ colored with $C=\{c_1,\dots,c_k\}$
\Ensure $P=(p_1,p_2,\dots, p_k,p_1)$ as the approximated polygon with the shortest perimeter that encompasses all distinct colors.

\State $R \gets $ Minimum Color-Spanning Circle on $S$ and $C$
\State $MSP \gets $ The set of points inside $R$
\State $MSP \gets $ Remove duplicate colors from MSP
%\Comment{ Assume that $c_0$ is $t_i$ }
\State $P \gets ()$

\While{$MSP$ is not empty } 

\State $L \gets $ Points on the boundary of the convex hull of $MSP$
\State $P \gets P.L$  \Comment{ Concatenation of $P$ and $L$}
\State $MSP \gets  MSP-L$ \Comment{ Remove $L$ from $MSP$}

\EndWhile
\State $P \gets P.P(1)$ \Comment{ Add first point of $P$ to $P$}

\end{algorithmic}
\end{algorithm}

In Algorithm~\ref{alg:2}, we begin by employing the algorithm outlined in~\cite{abellanas2001farthest} to compute $R$, representing the minimum color-spanning circle of the colored points $S$, within an efficient runtime of $O(k^3n\log{n})$. Subsequently, let $MSP$ denote the set of points within the circle $R$, each of which contains at least one point of each color. To ensure we avoid visiting the same color multiple times, we eliminate points with duplicate colors from $MSP$. As a result, the cardinality of $MSP$ is $\lvert MSP \lvert=k$.

In the next steps, we calculate the layers of onion peeling on $MSP$. For each layer $L$, we traverse the points of $L$ in a clockwise manner, commencing with the leftmost point. In the worst-case scenario, where the maximum onion depth is $\frac{k}{3}$, the time complexity of the while loop is $O(k^2\log{k})$. Subsequently, Theorem~\ref{thm:2} accompanied by a lemma demonstrates that the approximation factor of algorithm~\ref{alg:2} amounts to $\frac{2\pi r}{3}$. According to Pick's theorem, the lower bound for the area of each polygon within the set of $n$ points is $\frac{n}{2}-1$. Lemma~\ref{lm:1} provides a lower bound for the perimeter of these polygons.

%Based on Pick's theorem, the minimum area of polygon on a set of $n$ points is $\frac{n}{2}-1$. Lemma~\ref{lm:1}

%We remove the points with the same color~~~~~~~~~~~~~~~

% based on ??? fomula, the minimum area of polygon on $n$ points is $???$. The folloing lemma shows that the minimum perimeter is greater than $n$.

\begin{lemma}
\label{lm:1}
Given a set $S$ of $n$ points on the grid, let $P$ be the polygon with the minimum perimeter that crosses all the points of $S$. Then, the perimeter of $P$ is greater than or equal to $n$.
\end{lemma}

\begin{proof}
If $P$ is a polygon with $n$ edges, and we assume by contradiction that the perimeter of $P$ is less than $n$, then it would imply that the length of an edge of $P$ is less than 1. However, since the points of $S$ are on the grid, it's impossible for the distance between two points to be less than 1.
\end{proof}

\begin{theorem}
\label{thm:2}
The approximation factor of Algorithm~\ref{alg:2} is $\frac{2\pi r}{3}$, where $r$ is the radius of the smallest color-spanning circle on $S$.
\end{theorem}

\begin{proof}
%Since the points of $MSP$ are inside the circle $R$, for each layer $L$ of onion peeling, the length of $L$ is less than or equal to $2\pi r$. 
Given the points of $MSP$ inside the circle $R$, each layer $L$ of onion peeling has a length less than or equal to the circle's circumference, $2\pi r$.
In the worst case, if we have $\frac{k}{3}$ layers, the perimeter of $P$ will be less than or equal to $\frac{k}{3}2 \pi r$. 
On the other hand, if we assume that $OPT$ represents the perimeter of the optimal polygon $P^*$ on $MSP$, as per Lemma~\ref{lm:1}, $OPT$ is greater than or equal to $k$. Therefore, the perimeter of $P$ will be less than or equal to $\frac{2 \pi r}{3} * OPT$
\end{proof}

\section{Implementation}
\label{sec:3}

We have implemented both algorithms~\ref{alg:1} and~\ref{alg:2} in Matlab and executed them on multiple datasets. The results shown in Fig~\ref{fig:2} pertain to sets of 50 and 70 random points, each with 5 and 7 distinct colors, respectively. It's important to note that due to its time-consuming nature, Algorithm~\ref{alg:1} was not viable for running on larger datasets. Subsequently, Fig~\ref{fig:3} demonstrates the results of Algorithm~\ref{alg:2} when applied to sets of 500 and 1000 random points, with 10 and 20 distinct colors, respectively.

\begin{figure}[h]%
\centering
  \includegraphics[width=0.7\linewidth]{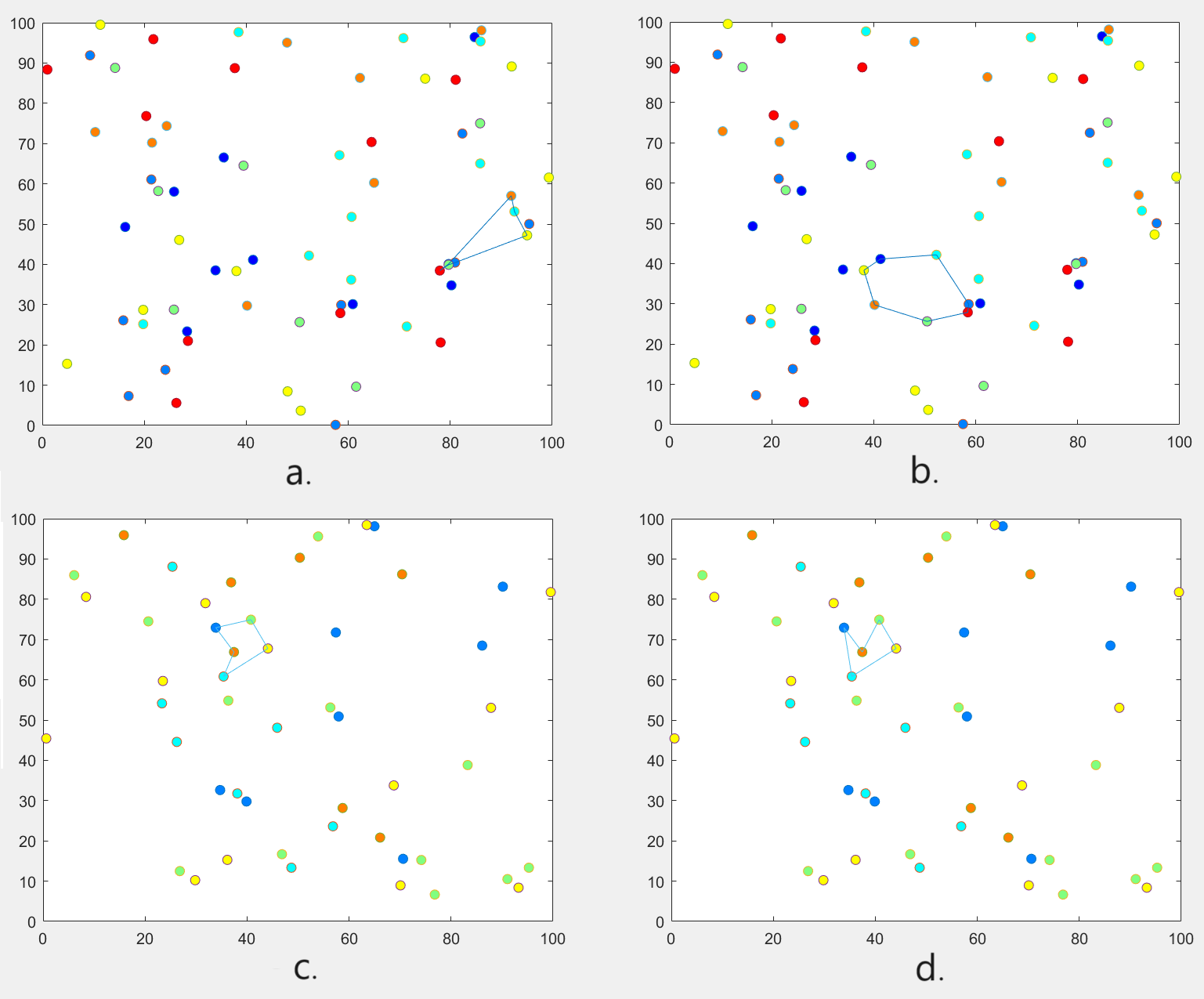}
  \caption{Figures (a) and (b) display the results of both the exact and approximate algorithms for the problem with $n=70$ and $k=7$, while figures (c) and (d) show the results for the problem with $n=50$ and $k=5$.}
  \label{fig:2}
\end{figure}

\begin{figure}[h]%
\centering
  \includegraphics[width=\linewidth]{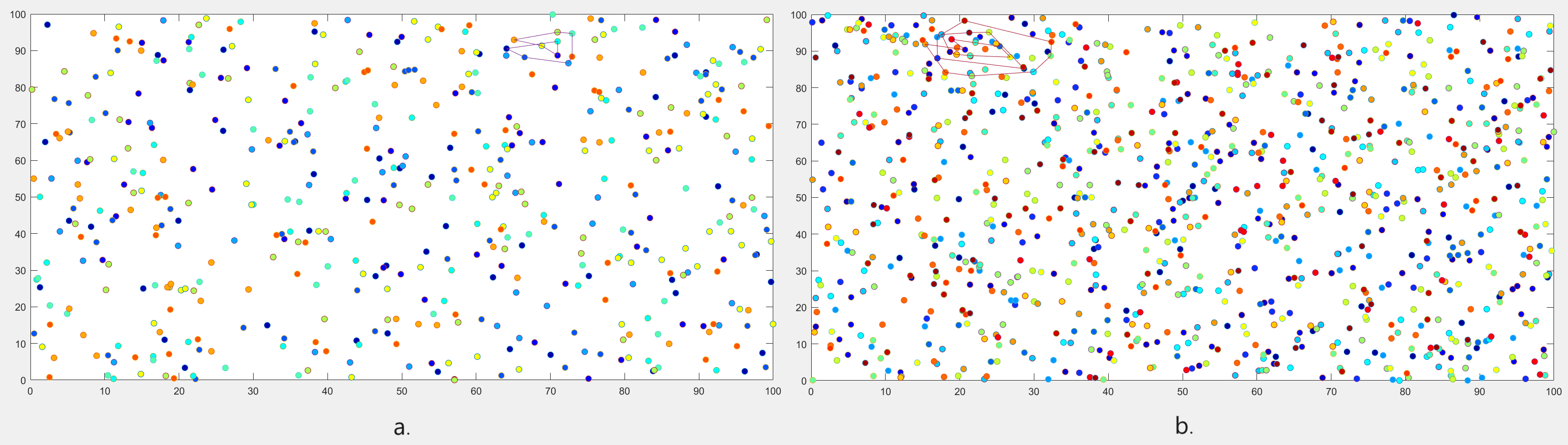}
  \caption{The results of the approximation algorithm are presented for (a) a set of 500 points with 10 distinct colors and (b) a set of 1000 points with 20 distinct colors..}
  \label{fig:3}
\end{figure}

As previously discussed, Algorithm~\ref{alg:1} has a time complexity of $O(\frac{n}{k}^k)$ in the worst case, while Algorithm~\ref{alg:2} operates at a time complexity of $O(k^3n\log{n})$. The comparison of the execution times for these algorithms is presented in Table~\ref{tbl:1}, Fig~\ref{fig:4} and Fig~\ref{fig:5}. Table~\ref{tbl:1} provides the execution times and the corresponding perimeters for both the exact algorithm (Exact) and the approximation algorithm (APX) across various values of $n$ and $k$.

\begin{table}[]
\centering
\resizebox{\textwidth}{!}{
\begin{tabular}{lll|l|l|l|l|l|}
\cline{4-8}
                                            &                                                 &       & $k=4$      & $k=5$      & $k=6$      & $k=7$      & $k=8$      \\ \hline
\multicolumn{1}{|l|}{\multirow{4}{*}{$n=10$}} & \multicolumn{1}{l|}{\multirow{2}{*}{Time}}      & Exact & 1.665056 & 0.815696 & 0.673525 & 0.191644 & 2.578316 \\ \cline{3-8} 
\multicolumn{1}{|l|}{}                      & \multicolumn{1}{l|}{}                           & APX   & 0.269556 & 0.086585 & 0.074245 & 0.107769 & 0.087482 \\ \cline{2-8} 
\multicolumn{1}{|l|}{}                      & \multicolumn{1}{l|}{\multirow{2}{*}{Perimeter}} & Exact & 134.9572 & 66.407   & 170.7678 & 162.7868 & 100.9232 \\ \cline{3-8} 
\multicolumn{1}{|l|}{}                      & \multicolumn{1}{l|}{}                           & APX   & 134.9572 & 66.407   & 170.7678 & 162.7868 & 105.152  \\ \hline
\multicolumn{1}{|l|}{\multirow{4}{*}{$n=15$}} & \multicolumn{1}{l|}{\multirow{2}{*}{Time}}      & Exact & 3.242028 & 3.53374  & 4.085118 & 4.143144 & 5.221325 \\ \cline{3-8} 
\multicolumn{1}{|l|}{}                      & \multicolumn{1}{l|}{}                           & APX   & 0.101243 & 0.089292 & 0.093642 & 0.091797 & 0.109786 \\ \cline{2-8} 
\multicolumn{1}{|l|}{}                      & \multicolumn{1}{l|}{\multirow{2}{*}{Perimeter}} & Exact & 95.94496 & 125.3841 & 63.1264  & 217.3375 & 218.7068 \\ \cline{3-8} 
\multicolumn{1}{|l|}{}                      & \multicolumn{1}{l|}{}                           & APX   & 95.94496 & 147.0142 & 63.1264  & 262.5696 & 251.6334 \\ \hline
\multicolumn{1}{|l|}{\multirow{4}{*}{$n=20$}} & \multicolumn{1}{l|}{\multirow{2}{*}{Time}}      & Exact & 7.94314  & 11.63913 & 11.13294 & 18.85075 & 12.58581 \\ \cline{3-8} 
\multicolumn{1}{|l|}{}                      & \multicolumn{1}{l|}{}                           & APX   & 0.139876 & 0.12601  & 0.12412  & 0.125416 & 0.133383 \\ \cline{2-8} 
\multicolumn{1}{|l|}{}                      & \multicolumn{1}{l|}{\multirow{2}{*}{Perimeter}} & Exact & 55.81232 & 124.1199 & 104.1315 & 112.4648 & 206.2783 \\ \cline{3-8} 
\multicolumn{1}{|l|}{}                      & \multicolumn{1}{l|}{}                           & APX   & 78.39269 & 127.2698 & 158.5983 & 112.4648 & 288.5908 \\ \hline
\multicolumn{1}{|l|}{\multirow{4}{*}{$n=25$}} & \multicolumn{1}{l|}{\multirow{2}{*}{Time}}      & Exact & 14.82877 & 28.86919 & 62.3261  & 57.34854 & 84.27023 \\ \cline{3-8} 
\multicolumn{1}{|l|}{}                      & \multicolumn{1}{l|}{}                           & APX   & 0.203443 & 0.174234 & 0.179336 & 0.157593 & 0.153333 \\ \cline{2-8} 
\multicolumn{1}{|l|}{}                      & \multicolumn{1}{l|}{\multirow{2}{*}{Perimeter}} & Exact & 85.45303 & 93.87832 & 91.25952 & 131.4472 & 110.7979 \\ \cline{3-8} 
\multicolumn{1}{|l|}{}                      & \multicolumn{1}{l|}{}                           & APX   & 101.6047 & 94.74419 & 105.7901 & 167.6262 & 132.0743 \\ \hline
\multicolumn{1}{|l|}{\multirow{4}{*}{$n=30$}} & \multicolumn{1}{l|}{\multirow{2}{*}{Time}}      & Exact & 30.30081 & 111.3685 & 168.3525 & 171.341  & 383.6654 \\ \cline{3-8} 
\multicolumn{1}{|l|}{}                      & \multicolumn{1}{l|}{}                           & APX   & 0.377235 & 0.388495 & 0.362932 & 0.30269  & 0.329382 \\ \cline{2-8} 
\multicolumn{1}{|l|}{}                      & \multicolumn{1}{l|}{\multirow{2}{*}{Perimeter}} & Exact & 43.08548 & 46.68444 & 71.69952 & 116.0164 & 186.5542 \\ \cline{3-8} 
\multicolumn{1}{|l|}{}                      & \multicolumn{1}{l|}{}                           & APX   & 45.94552 & 68.19333 & 111.9154 & 187.6612 & 233.8898 \\ \hline
\multicolumn{1}{|l|}{\multirow{4}{*}{$n=35$}} & \multicolumn{1}{l|}{\multirow{2}{*}{Time}}      & Exact & 73.85159 & 58.34555 & 455.2198 & 589.9313 & 622.9517 \\ \cline{3-8} 
\multicolumn{1}{|l|}{}                      & \multicolumn{1}{l|}{}                           & APX   & 0.539233 & 0.433459 & 0.534906 & 0.512461 & 0.433698 \\ \cline{2-8} 
\multicolumn{1}{|l|}{}                      & \multicolumn{1}{l|}{\multirow{2}{*}{Perimeter}} & Exact & 36.97732 & 112.902  & 77.66447 & 85.63003 & 120.2107 \\ \cline{3-8} 
\multicolumn{1}{|l|}{}                      & \multicolumn{1}{l|}{}                           & APX   & 36.97732 & 169.5268 & 132.1911 & 85.63003 & 193.1477 \\ \hline
\multicolumn{1}{|l|}{\multirow{4}{*}{$n=40$}} & \multicolumn{1}{l|}{\multirow{2}{*}{Time}}      & Exact & 106.644  & 387.588  & 580.1023 & 2345.168 & 3685.009 \\ \cline{3-8} 
\multicolumn{1}{|l|}{}                      & \multicolumn{1}{l|}{}                           & APX   & 0.743369 & 0.719214 & 0.687392 & 0.651942 & 0.705154 \\ \cline{2-8} 
\multicolumn{1}{|l|}{}                      & \multicolumn{1}{l|}{\multirow{2}{*}{Perimeter}} & Exact & 41.00075 & 38.47183 & 102.8376 & 103.4379 & 89.85302 \\ \cline{3-8} 
\multicolumn{1}{|l|}{}                      & \multicolumn{1}{l|}{}                           & APX   & 52.00889 & 41.79108 & 125.3398 & 189.4559 & 152.3103 \\ \hline
\end{tabular}
}
\caption{Resuls of some samples for varying values of $n$ and $k$}
  \label{tbl:1}
\end{table}

\begin{figure}[h]%
\centering
  \includegraphics[width=\linewidth]{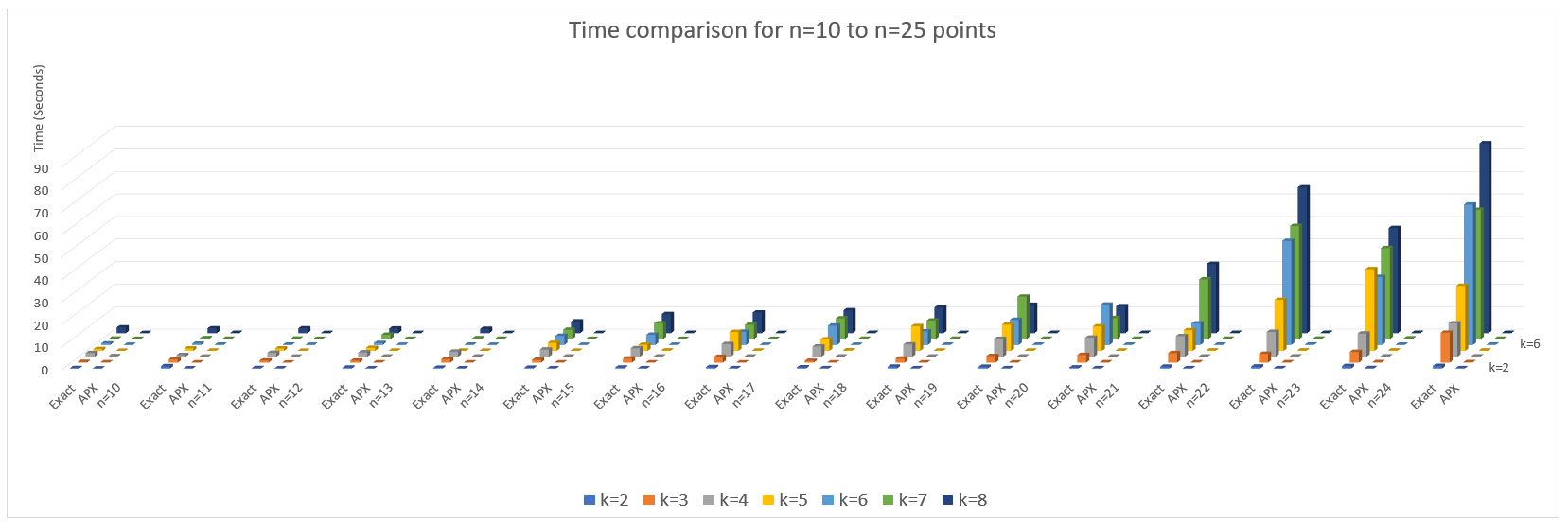}
  \caption{Comparison of the algorithms' execution time for $n$ ranging from 10 to 25 points and $k$ spanning from 2 to 8 colors.}
  \label{fig:4}
\end{figure}

\begin{figure}[h]%
\centering
  \includegraphics[width=\linewidth]{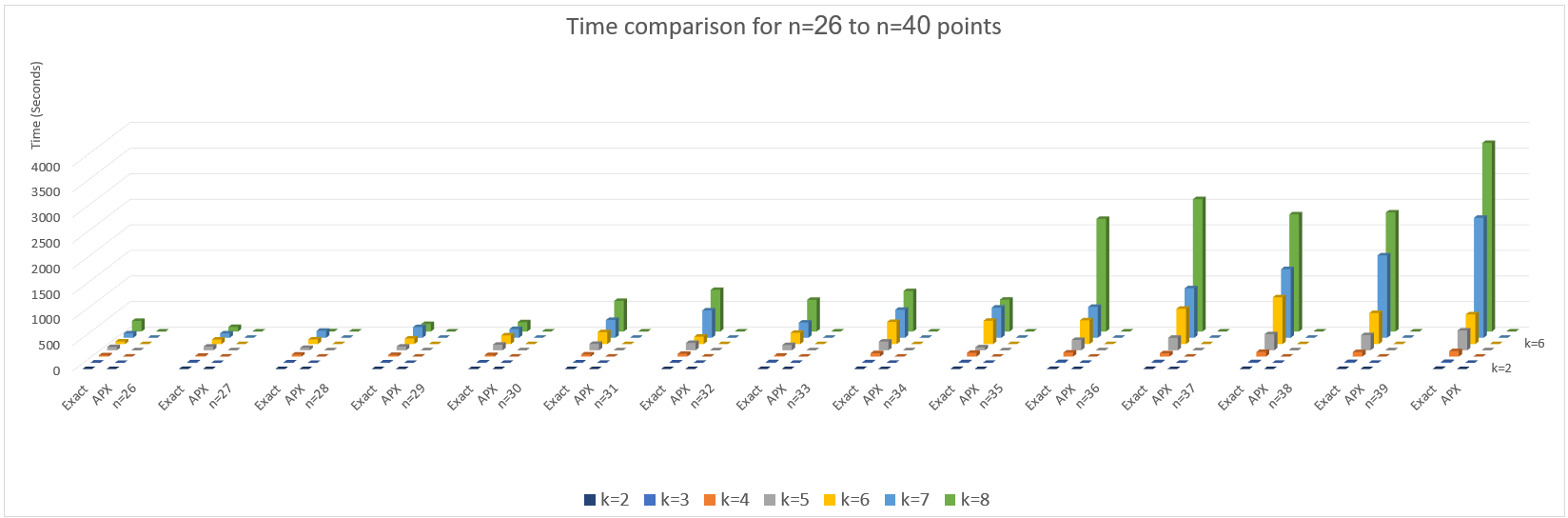}
  \caption{Comparison of the algorithms' execution time for $n$ ranging from 26 to 40 points and $k$ spanning from 2 to 8 colors.}
  \label{fig:5}
\end{figure}

\section{Conclusion}
\label{sec:4}
In this paper, the Colored Points TSP is defined as a variation of the TSP where points are colored with a set of distinct colors, and the objective is to find the minimum cycle to visit each color once. The NP-completeness of this problem is proven, and two exact and approximate algorithms are presented to solve this problem. The factor of the approximation algorithm is $\frac{2\pi r}{3}$, where $r$ is the radius of the smallest color-spanning circle of the points. The presented algorithms have been implemented, and numerical results have been obtained on random datasets. The efficiency of the approximation algorithm is demonstrated by the results.

%
%\section*{Authors' contributions}
%All authors wrote the main manuscript text, prepared all figures, designed and implemented the presented algorithm and analyzed the results. All authors reviewed and approved the manuscript.
%
%\section*{Funding}
%
%The research of the first and second authors are partially supported by the University of Kashan under grant numbers 991449/3 and 1143902/2, respectively.
%
%\section*{Availability of data and materials}
%not applicable



  \bibliographystyle{elsarticle-num} 
  \bibliography{mybibfile}

\begin{thebibliography}{10}
\expandafter\ifx\csname url\endcsname\relax
  \def\url#1{\texttt{#1}}\fi
\expandafter\ifx\csname urlprefix\endcsname\relax\def\urlprefix{URL }\fi
\expandafter\ifx\csname href\endcsname\relax
  \def\href#1#2{#2} \def\path#1{#1}\fi

\bibitem{laporte1992traveling}
G.~Laporte, The traveling salesman problem: An overview of exact and
  approximate algorithms, European Journal of Operational Research 59~(2)
  (1992) 231--247.

\bibitem{harvey1974chromatic}
M.~E. Harvey, R.~T. Hocking, J.~R. Brown, The chromatic traveling-salesmen
  problem and its application to planning and structuring geographic space,
  Geographical Analysis 6~(1) (1974) 33--52.

\bibitem{xiong2007colorful}
Y.~Xiong, B.~Golden, E.~Wasil, The colorful traveling salesman problem,
  Extending the Horizons: Advances in Computing, Optimization, and Decision
  Technologies (2007) 115--123.

\bibitem{tresoldi2010solving}
E.~Tresoldi, R.~Wolfler~Calvo, S.~Borne, Solving the multicolor tsp, in: 24th
  European Conference on Operational Research (EURO XXIV). Lisbon, Portugal,
  2010.

\bibitem{li2014colored}
J.~Li, M.~Zhou, Q.~Sun, X.~Dai, X.~Yu, Colored traveling salesman problem, IEEE
  transactions on cybernetics 45~(11) (2014) 2390--2401.

\bibitem{meng2021colored}
X.~Meng, J.~Li, M.~Zhou, A colored traveling salesman problem with varying city
  colors, Discrete Dynamics in Nature and Society 2021 (2021) 1--14.

\bibitem{couetoux2010labeled}
B.~Cou{\"e}toux, L.~Gourves, J.~Monnot, O.~A. Telelis, Labeled traveling
  salesman problems: Complexity and approximation, Discrete Optimization
  7~(1-2) (2010) 74--85.

\bibitem{sylvester1857question}
J.~J. Sylvester, A question in the geometry of situation, Quarterly Journal of
  Pure and Applied Mathematics 1~(1) (1857) 79--80.

\bibitem{megiddo1983linear}
N.~Megiddo, Linear-time algorithms for linear programming in r\^{}3 and related
  problems, SIAM journal on computing 12~(4) (1983) 759--776.

\bibitem{abellanas2001farthest}
M.~Abellanas, F.~Hurtado, C.~Icking, R.~Klein, E.~Langetepe, L.~Ma, B.~Palop,
  V.~Sacrist{\'a}n, The farthest color voronoi diagram and related problems,
  in: Abstracts 17th European Workshop Comput. Geom, Freie Universit{\"a}t
  Berlin, 2001, pp. 113--116.

\bibitem{abellanas2001smallest}
M.~Abellanas, F.~Hurtado, C.~Icking, R.~Klein, E.~Langetepe, L.~Ma, B.~Palop,
  V.~Sacrist{\'a}n, Smallest color-spanning objects, in: Algorithms—ESA 2001:
  9th Annual European Symposium {\AA}rhus, Denmark, August 28--31, 2001
  Proceedings 9, Springer, 2001, pp. 278--289.

\bibitem{acharyya2022minimum}
A.~Acharyya, R.~K. Jallu, V.~Keikha, M.~L{\"o}ffler, M.~Saumell, Minimum color
  spanning circle of imprecise points, Theoretical Computer Science 930 (2022)
  116--127.

\end{thebibliography}


%\begin{thebibliography}{00}
%%% Text of bibliographic item
%
%%\bibitem{}
%
%\end{thebibliography}

\end{document}